\documentclass[11pt]{article}
\usepackage{fullpage}
\usepackage{subfig}
\usepackage[usenames,dvipsnames]{xcolor}

\usepackage[linktocpage=true,
	pagebackref=true,colorlinks,
	linkcolor=BrickRed,citecolor=blue,
	bookmarks,bookmarksopen,bookmarksnumbered]
	{hyperref}

	\usepackage{latexsym,graphicx,epsfig,color}
\usepackage{amsfonts,amssymb,amsmath,amsthm,amstext}
\usepackage{lmodern}
\usepackage{enumitem}
\setitemize{itemsep=2pt,topsep=0pt,parsep=0pt}
\usepackage{url,setspace}
\usepackage{multirow}
\usepackage{rotating}
\usepackage{makeidx}
\usepackage{tikz}
\usetikzlibrary{arrows}
\usetikzlibrary{arrows.meta}
\usetikzlibrary{shapes}
\usetikzlibrary{backgrounds}
\usetikzlibrary{positioning}
\usetikzlibrary{decorations.markings}
\usetikzlibrary{patterns}
\usetikzlibrary{calc}
\usetikzlibrary{fit}
\usetikzlibrary{snakes}

\usepackage{accents}
\usepackage{xspace}
\usepackage{algorithm,algpseudocode}
\usepackage{bm}
\usepackage{changepage} 	
\usepackage{thmtools,thm-restate}

\newcommand{\IGNORE}[1]{}
\allowdisplaybreaks

\usetikzlibrary{decorations.markings}
\usetikzlibrary{arrows}
\tikzstyle{block}=[draw opacity=0.7,line width=1.4cm]
\tikzstyle{graphnode}=[circle, draw, fill=black!20, inner sep=0pt, minimum width=6pt]
\tikzstyle{point}=[circle, draw, fill=black!30, inner sep=0pt, minimum width=1pt]
\tikzstyle{input}=[rectangle, draw, fill=black!75,inner sep=3pt, inner ysep=3pt, minimum width=4pt]
\tikzstyle{unmatched}=[graphnode,fill=black!0]
\tikzstyle{shaded}=[graphnode,fill=black!20]
\tikzstyle{matched}=[graphnode,fill=black!100]  	
\tikzstyle{matching} = [ultra thick]
\tikzset{
    >=stealth',
    pil/.style={
           ->,
           thick,
           shorten <=2pt,
           shorten >=2pt,}
}
\tikzset{->-/.style={decoration={
  markings,
  mark=at position .5 with {\arrow{>}}},postaction={decorate}}}

\makeindex

\DeclareMathOperator{\argmax}{arg max}

\newtheorem{theorem}{Theorem}
\newtheorem{lemma}{Lemma}[section]
\newtheorem{claim}[lemma]{Claim}
\newtheorem{proposition}[lemma]{Proposition}

\newtheorem{observation}[lemma]{Observation}

\theoremstyle{definition}

\usepackage{thmtools,thm-restate} 

\usepackage{mdframed}
\newtheorem*{mdresult}{Main Result}
\newenvironment{result}{\begin{mdframed}[backgroundcolor=lightgray!40,topline=false,rightline=false,leftline=false,bottomline=false,innertopmargin=2pt]\begin{mdresult}}{\end{mdresult}\end{mdframed}}

\newcommand{\E}{\mathbb{E}}

\newcommand{\cupdot}{\mathbin{\mathaccent\cdot\cup}}

\def \reals {\mathbb{R}}

\newcommand{\one}{\mathbf{1}\xspace}
\newcommand{\poly}{\operatorname{poly}}

\newcounter{note}[section]

\usepackage{bm}
\usepackage[most]{tcolorbox}

\newcommand{\paren}[1]{\ensuremath{\left(#1\right)}\xspace}
\newcommand{\alg}{\ensuremath{\mathcal{A}}\xspace}

\newcommand{\bq}{\mathbf{q}}

\newcommand{\OPT}{\ensuremath{\mathsf{OPT}}}
\newcommand{\Util}{\ensuremath{\mathsf{Utility}}}
\newcommand{\Rev}{\ensuremath{\mathsf{Rev}}}
\newcommand{\Welfare}{\ensuremath{\mathsf{Welfare}}}

\newcommand{\Ai}[1]{\ensuremath{A^{(#1)}}}

\newcommand{\price}{\ensuremath{p}}

\newcommand{\bprice}{\bm{\price}}
\newcommand{\bpricei}[1]{\ensuremath{\bprice^{(#1)}}}

\newcommand{\mech}{\ensuremath{\textnormal{\texttt{BinarySearchMechanism}}}\xspace}
\newcommand{\spmech}{\ensuremath{\textnormal{\texttt{FixedPriceAuction}}}\xspace}

\usepackage[compact]{titlesec}
\title{Improved Truthful Mechanisms for 
 Subadditive \\ Combinatorial Auctions: Breaking the Logarithmic Barrier}

\author{%
Sepehr Assadi\thanks{(sepehr.assadi@rutgers.edu) Department of Computer Science, Rutgers University.}
\and
Thomas Kesselheim\thanks{(thomas.kesselheim@uni-bonn.de)
 Institute of Computer Science,	    University of Bonn.}
 \and
 {Sahil Singla}\thanks{(singla@cs.princeton.edu)
    Department of Computer Science at
        Princeton University and
        School of Mathematics at 
        Institute for Advanced Study.        Supported in part by the Schmidt Foundation.}
}

\date{}

\begin{document}

\setlength{\abovedisplayskip}{2pt}
\setlength{\belowdisplayskip}{2pt}

\maketitle

\begin{abstract}
\medskip

We present a computationally-efficient truthful mechanism for combinatorial auctions with \emph{subadditive} bidders that achieves an $O((\log\!\log{m})^3)$-approximation to the maximum welfare in expectation using $O(n)$ demand queries; here $m$ and $n$
are the number of items and bidders, respectively. This breaks the longstanding logarithmic barrier for the problem dating back to the $O(\log{m}\cdot\log\!\log{m})$-approximation mechanism of Dobzinski from 2007. Along the way, we also improve and considerably simplify the state-of-the-art  mechanisms
 for submodular bidders. 

\end{abstract}

\bigskip

\tableofcontents

\clearpage



\section{Introduction}\label{sec:intro}

In a combinatorial auction, $m$ items in $M$ are to be allocated to $n$ bidders in $N$ who have valuations for bundles of items. 
Each bidder $i \in N$ has a valuation function $v_i\colon: 2^M \rightarrow \mathbb{R}^{+}$ that describes their value $v_i(S)$ for any bundle $S$ of items. The goal is to design a mechanism that finds an allocation $A = (A_1,\ldots,A_n)$ of the items 
to the bidders so as to maximize the \emph{social welfare} defined as $\Welfare(A) := \sum_{i \in N} v_i(A_i)$. We need the 
mechanisms to be \emph{computationally efficient} in that they run in $\poly(m,n)$ time given proper query access to the valuations of bidders, namely, value queries and demand queries (see~\S\ref{sec:prelim} for definitions).  
Mechanisms should also take into account the strategic behavior of the bidders. A mechanism in which the dominant strategy of each bidder is to reveal their true valuation in
response to given queries is called truthful. For randomized mechanisms, we consider \emph{universally
truthful} mechanisms which are simply distributions over truthful mechanisms.

A central question of Algorithmic Mechanism Design is to understand the inherent clash between computational-efficiency and truthfulness for
 combinatorial auctions (see, e.g.~\cite{DobzinskiNS06,AbrahamBDR12,FotakisKV17,MualemN08,DughmiV11,Dobzinski07,Dobzinski16,AS-FOCS19,AKSW-STOC20}). 
 On one hand, the celebrated VCG mechanism of Vickrey, Clarke, and Groves \cite{Vickrey61,Clarke71,Groves73} is a truthful mechanism that returns the welfare maximizing allocation but requires exponential time 
 for most classes of valuations. On the other hand, constant factor {approximation} algorithms exist for many interesting classes of valuations, but they are not truthful. 
 Are there mechanisms that are both computationally efficient and truthful while still matching  the performance of these algorithms? 
 
 A particularly interesting case of this question is when the bidders' valuations are \emph{subadditive}\footnote{A valuation $v(\cdot)$ is subadditive if for all $A,B \subseteq M$, we have  $v(A \cup B) \leq v(A) + v(B)$.}. 
 This captures the property that items are not complementary in that they will not become more valuable together than apart, a natural and important property in many contexts. Moreover, 
 subadditive valuations encompass various other valuation classes of interest  as special cases, e.g.,  XOS, submodular, and gross substitute valuations (see~\cite{LehmannLN06} for definitions and significance of these classes in the context of combinatorial auctions). 
Finding a welfare maximizing allocation with subadditive bidders is not possible in polynomial time~\cite{DobzinskiNS05}, and thus VCG is not a computationally efficient mechanism for this problem. 
However, from a purely algorithmic perspective, there is a poly-time $2$-approximation algorithm for subadditive combinatorial auctions due to Feige~\cite{Feige06}, improving upon the $O(\frac{\log{m}}{\log\!\log{m}})$ approximation algorithm of 
Dobzinski~et.al.~\cite{DobzinskiNS05}. 

The state-of-the-art truthful and computationally efficient mechanisms for subadditive combinatorial auctions are lagging considerably behind. The first such mechanism  was 
obtained by Dobzinski, Nisan, and Schapira and achieves an $O(\sqrt{m})$-approximation~\cite{DobzinskiNS05}. This was soon after improved by Dobzinski in~\cite{Dobzinski07} to an $O(\log{m}\log\!\log{m})$ approximation.  
Breaking this logarithmic barrier however has remained elusive for way over a decade now.  
 In the meantime, the special case of this problem with submodular (and even XOS) valuations has witnessed tremendous progress: a long line of
work in~\cite{DobzinskiNS06,Dobzinski07,KrystaV12,Dobzinski16,AS-FOCS19} culminated in an $O(\sqrt{\log{m}})$-approximation mechanism of Dobzinski~\cite{Dobzinski16} that finally broke
this logarithmic barrier, and was subsequently improved by the first and the last authors to an $O((\log\!\log{m})^3)$-approximation~\cite{AS-FOCS19}. These results  fall short of addressing subadditive valuations 
due to the known $\Theta(\log{m})$-gap for approximating subadditive valuations via XOS valuations\footnote{A valuation $v'(\cdot)$ $\gamma$-approximates another valuation $v(\cdot)$ if for all $S \subseteq M$, we have $v(S) \leq v'(S) \leq \gamma \cdot v(S)$.}. Indeed, to quote~\cite{Dobzinski16}: ``A logarithmic factor
loss seems inevitable with this approach [i.e., combining~\cite{Dobzinski07} with~\cite{Dobzinski16} (and~\cite{AS-FOCS19} also)]. Breaking the logarithmic barrier also for subadditive
valuations looks challenging''.

\paragraph{Our Result.} We present a new mechanism  that breaks the logarithmic barrier on the approximation ratio of subadditive combinatorial auctions, all the way to a $\poly\!\log\!\log{\!(m)}$ factor. 

\begin{result}
There exists a universally truthful mechanism for combinatorial auctions with
subadditive bidders that achieves an $O((\log\!\log{m})^3)$-approximation to the social
welfare in expectation using at most $n$ 
value and demand queries.
\end{result}

Our result  matches the state-of-the-art for the special case of submodular (and XOS) valuations~\cite{AS-FOCS19} and unifies the previous two lines of work on subadditive and submodular valuations. 
Furthermore, our approach considerably simplifies the previously best mechanism of~\cite{AS-FOCS19}, and improves it from $O((\log\!\log{m})^3)$-approximation to $O((\log\!\log{m})^2)$-approximation (Theorem~\ref{thm:finalMain}). 
 
 \paragraph{Our Techniques.} 
 Our starting point is the prior work on the existence of ``good'' prices for XOS (hence also for submodular) \cite{FeldmanGL15,Dobzinski07} and subadditive valuations \cite{DKL-FOCS20}; these prices are such that 
 a \emph{fixed-price auction} (FPA) on these prices results in an $O(1)$-approximation of welfare for XOS and $O(\log\!\log{m})$-approximation for subadditive bidders. 
 The catch with these prices is that one needs prior knowledge of the (probability distributions of the) valuations, which we cannot assume. The goal in all truthful mechanisms  in this line of work~\cite{DobzinskiNS06,Dobzinski07,KrystaV12,Dobzinski16,AS-FOCS19} is  therefore to obtain ``some estimates'' on these prices and then allocate items using a FPA with these estimates. 
 This is similarly the case for our mechanism. 
 
 For now, let us assume that these good prices belong to the set $\{1,2,4,8,\ldots,\poly(m)\}$ of $O(\log{m})$ different prices; this assumption captures the essence of the problem and can be removed using standard ideas. 
To find a price for every item, our mechanism narrows down the choice of candidate prices for each item in multiple rounds. In particular, in each of these rounds,  either the upper or the lower half of all candidate prices for an item are eliminated 
similar to a binary search. We do this by picking for each item, the price right in the middle of its remaining candidate prices, and then running one FPA on all items but involving only a subset of bidders. This FPA should mostly be seen as a ``practice run''. For all items that are sold, we eliminate the lower half of the prices, and for the remaining items, we eliminate the upper half. After $O(\log\!\log m)$ rounds, only one price remains per item, which is then used in the last round's FPA. There is one further caveat though: Any FPA to narrow down the set of candidate prices is not necessarily a practice run. With a certain probability, they determine already the final allocation. This is necessary to compensate that the mechanism may not necessarily
be able to ``learn'' the correct prices. By ensuring that each bidder takes part in only one of these FPAs, the mechanism is immediately truthful because no bidder can actually influence their own price menu.

Our mechanism can be considered as a simplified version of the price-learning mechanism in~\cite{AS-FOCS19}, which involved using multiple prices for an item in each round to learn the final prices. 
Our analysis approach, however, is entirely different from the one in \cite{AS-FOCS19}. It is significantly simpler while also improving the approximation guarantee for XOS bidders. But more importantly, it allows us to work with the prices 
guaranteed for subadditive valuations in~\cite{DKL-FOCS20}, which are ``less robust'' than the ones for XOS valuations\footnote{For instance, XOS prices have the  property that 
even if only a subset of items are assigned their correct prices, a FPA
results in a welfare at least as large as the prices of these correctly-priced items. Such a guarantee is crucially used in~\cite{AS-FOCS19} as not all prices can be learned correctly and thus one needs to be able to consider only the items with correctly learned prices. However, 
subadditive  prices from~\cite{DKL-FOCS20} do not necessarily satisfy this guarantee.}.

The analysis in \cite{AS-FOCS19} tracks the progress of the mechanism in every round. It is shown that in each round, the mechanism can either refine the learned prices significantly
for most items (``Learnable''), or the intermediate FPAs already give a high-welfare allocation (``Allocatable''). Thus, one way or another, the mechanism guarantees 
a sufficiently large welfare. Unfortunately, maintaining such a guarantee for the case of subadditive bidders does not seem possible due to the different nature of prices for XOS versus subadditive valuations. 

Our new analysis gives \emph{no guarantee} for any particular round. Instead, we lower-bound the the expected utility of each bidder plus the revenue obtained from selling the optimal items for this bidder in the final outcome (not necessarily to this bidder)---it is easy to see that sum of this quantity for all bidders lower-bounds the welfare. At the same time, we show that this quantity is not much less than the optimal welfare contribution of this bidder.
To this end, we trace every item individually through the rounds. Effectively, one of these cases has to take place: Maybe the item reaches its intended final price, which means it will contribute in the last round's FPA. Or maybe 
the intended price gets eliminated in some round, which means that the item is either not sold at a smaller price, or it is sold at a higher price than the intended one. If the item is not sold at a smaller price, we argue that the bidder who gets this item in the optimal allocation could buy it in the respective round (bidders are assigned to different rounds \emph{randomly}).  If the item is sold at a higher price, we argue that we obtain enough revenue from the item. Due to subadditivity, we are able to combine the contribution of all these different rounds and conclude the analysis. 

\paragraph{Further Related Work.} The question that whether or not computationally efficient and truthful mechanisms can match the performance of algorithms has been studied extensively in the literature. 
For mechanisms that only use value queries, the answer is \emph{no}: algorithms that use only $\poly(m,n)$ many value queries 
can only achieve an $m^{\Omega(1)}$-approximation~\cite{PapadimitriouSS08,Dobzinski11,DughmiV11,DobzinskiV12,DobzinskiV12,DanielySS15} even 
for submodular valuations, while $O(1)$-approximation algorithms are known for these valuations~\cite{LehmannLN06,Vondrak08,FeigeV06}. 
However, these impossibility results no longer apply to mechanisms that are allowed other types of queries, in particular, demand queries. This has led the researchers to study
the communication complexity of this problem that capture arbitrary queries~\cite{Nisan00,BlumrosenN02,DobzinskiNS05,NisanS06,DobzinskiV13,DobzinskiNO14,Dobzinski16b,Assadi17ca,BravermanMW17,EzraFNTW19,AKSW-STOC20}. 
The only known  separation between communication complexity of truthful mechanism vs. arbitrary algorithms was very recently established in~\cite{AKSW-STOC20}, building on a characterization of truthful mechanisms in~\cite{Dobzinski16b} (see also~\cite{BravermanMW17,EzraFNTW19} for earlier attempts based on~\cite{Dobzinski16b}). 


\section{Preliminaries}\label{sec:prelim}

\paragraph{Notation.} We denote by $N$ the set of $n$ bidders and by $M$ the set of $m$ items.   We use bold-face letters to denote vectors of prices and capital letters for allocations. For a price vector $\bprice$ and a set of items $M' \subseteq M$, we define $\bprice(M') := \sum_{e \in M'} \price_e$. 
Let $(S_1^\star, \ldots, S_n^\star)$ denote the optimal allocation that maximizes the welfare.  Note that this allocation partitions the items $M$.

\subsection{Subadditive Valuation Functions}\label{sec:valuations}

We make the standard assumption that valuation $v_i$ of each bidder $i$ is normalized, i.e., $v_i(\emptyset)=0$, and is monotone, i.e., $v_i(S) \leq v_i(T)$ for every $S \subseteq T \subseteq M$. 
We are interested in the case when bidders valuations are \emph{subadditive} and hence capture the notion of ``complement free'' valuations.  
A valuation  $v$ is subadditive if  $v(S \cup T)  \leq v(S) + v(T)$ for any $S,T \subseteq M$. 


\paragraph{Query access to valuations.} Since valuations have size exponential in $m$, a common assumption is that valuations are specified via certain queries instead, in particular, value queries and demand queries. 
A value query to valuation $v$ on bundle $S$ reveals the value of $v(S)$. A demand query specifies a price vector $\bprice$ on items and the
answer is the ``most demanded'' bundle under this pricing, i.e., a bundle $S \in \argmax_{S'} \{v(S')-\bprice(S')\}$.

\subsection{Good Prices for Subadditive Bidders} 

Our mechanism crucially uses  the following prices of~\cite{DKL-FOCS20}. 

\begin{lemma}[\cite{DKL-FOCS20}] \label{lem:DKL}
Given any set of items $U \subseteq M$ and valuation $v: 2^M \rightarrow \mathbb{R}^+$, there is a distribution $\lambda$ over item sets $S \subseteq U$ and prices $\bq^\star \in \reals_{\geq 0}^{|U|}$ such that for any $T \subseteq U$, we have
\[  \sum_{S \subseteq U} \lambda_S \Big( v(S\setminus T) - \bq^\star(S) \Big) ~\geq~ \frac{1}{\alpha} \cdot v(U)  - \bq^\star(T).
\]
Here, if valuation $v$ is subadditive then $\alpha = O(\log\!\log\!m)$ and if $v$ is XOS then $\alpha = O(1)$.
\end{lemma}

The primary challenge in using the prices of \cite{DKL-FOCS20} is that they may not be robust to losing a fraction of items. Indeed, our mechanism can at best hope to learn prices of only a fraction of the items, but Lemma~\ref{lem:DKL} loses additively for any such lost items. 

\subsection{A Fixed-Price Auction}\label{sec:fixed-price}

We use a standard fixed-price auction as a subroutine in our mechanism. For an \emph{ordered} set $N$ of bidders, $M$ of items, and a price vector $\bprice$, define $\spmech(N,M,\bprice)$  as follows. 
\begin{tcolorbox}[breakable]
	\underline{$\spmech(N,M,\bprice)$}
	\begin{enumerate}
		\item Iterate over the bidders $i$ of the ordered set $N$ in the given order:
		\begin{enumerate}[topsep=0pt]
			\item Allocate $A_i \in \argmax_{S \subseteq M} \{ v_i(S) - \bprice(S) \}$ to bidder $i$ and update $M \leftarrow M \setminus A_i$. 
		\end{enumerate}
		\item Return the allocation $A = (A_1,\ldots,A_n)$. 
	\end{enumerate}
\end{tcolorbox}

It is easy to see that $\spmech$ can be implemented using one demand query per bidder. Its truthfulness is also easy to check as  bidders have no influence on the pricing mechanism. 

\paragraph{Utility and revenue.} For a $\spmech$, we denote the \emph{utility} of any bidder $i\in N$ by $\Util_i := v_i(A_i) - \bprice(A_i)$. For any set $S\subseteq M$ of items, we denote the \emph{revenue} obtained by selling items in $S$ by $\Rev(S):= \sum_{e \in S} p_e \cdot \one_{e~\text{is allocated}}$. This implies that the total welfare $\sum_{i\in N} v_i(A_i)$ of a $\spmech$ equals $\sum_{i\in N} \Util_i + \Rev(M)$.



	\section{The Binary-Search Mechanism} \label{sec:main-mech}
	
	In this section, we design our mechanism under the assumption  that we have a rough estimate of the scale of valuation functions. Formally, let us assume that we are given an estimate $\psi$ for  the value $\max_i v_i(M)$.
	Simply asking the bidders for this value would lead to complex incentive issues, but the assumption can be removed using common sampling techniques from algorithmic mechanism design, which we give in \S\ref{sec:end-mech} for completeness. 
	
	The idea of our mechanism is to find a good price for every item. The knowledge of $\psi$ helps in the following way. Observe that pricing an item above $\psi$ will  be unreasonable since no bidder would be willing to buy it. At the same time, items that are only bought if they have a ``tiny'' value compared to $\psi$ are not important for social welfare. Such considerations lead to the observation that we can restrict our attention to a small set $B$ of $O(\log{m})$ \emph{candidate prices} (Observation~\ref{obs:roundedPrices}).
	
	Based on this set of candidate prices, we design our binary-search mechanism in \S\ref{sec:binary-search-on-prices}. For every item, we run a binary-search procedure of $\log_2 \lvert B \rvert = O(\log\!\log{m})$ rounds. In each round, half of the candidate prices are eliminated, which leaves us with a unique price for each item in the end.

	
	
	\subsection{Existence of Good Candidate Prices}
	\label{sec:cand-prices}
	
	As a first step, we observe that given $\psi$ it is enough to restrict our attention to a small number of candidate prices. More precisely, we derive a set of $O(\log m)$ prices, which are good in the sense that Lemma~\ref{lem:DKL} still applies approximately if we restrict to these prices. The reason is that neither very high nor very low prices are important, and that it is okay to round prices to powers of 2. The following observation follows easily from Lemma~\ref{lem:DKL}.
	
\begin{observation} \label{obs:roundedPrices}
For a given $\psi$, define  $B := \{0, 2^{-3 \lceil \log_2 m \rceil} \psi, 2^{1 - 3 \lceil \log_2 m \rceil} \psi, \ldots, 2^{-1} \psi\}$ to be a set of $3 \lceil \log_2 m \rceil$ \emph{candidate prices}. 
Now for any set of items $U \subseteq M$ and any  bidder $i \in N$, if $v_i(U) \leq \psi$ then there exist a distribution $\lambda$ over item sets $S \subseteq U$ and  prices  $\bq \in B^{|U|}$ 
such that for any $T \subseteq U$, 
\[  \sum_{S \subseteq U} \lambda_S \Big( v_i(S\setminus T) - \bq(S) \Big) ~ \geq ~ \frac{1}{\alpha} \cdot v_i(U) - 2 \bq(T) - \frac{\psi}{m^2}.
\]
Here, if valuation $v_i$ is subadditive then $\alpha = O(\log\!\log m)$ and if $v_i$ is XOS then $\alpha = O(1)$.
\end{observation}

Note that the statement requires $\psi$ to be at least $v_i(U)$. If $\psi$ is in the order of $v_i(U)$, then the negative term $\frac{\psi}{m^2}$ gets negligible. The higher $\psi$ becomes, the weaker is the bound.

\begin{proof}[Proof of Observation~\ref{obs:roundedPrices}]
We obtain $q_e$ by rounding $q^\star_e$ from Lemma~\ref{lem:DKL} to the next smaller number in $B$. Note that this is always possible because $0$ is contained in $B$. We can assume without loss of generality that $q^\star_e \leq v_i(U) \leq \psi$ for all $e$ that are contained in a set $S$ in the support of $\lambda$ because otherwise $v_i(S\setminus T) - q(S) < v_i(U) - v_i(U) = 0$ for any set $T$. This would mean that the contribution of $S$ to sum is always negative and one could shift the probability mass to the empty set. Overall, this means that for all relevant $e$, we have $q_e \leq q^\star_e \leq 2 q_e + \frac{\psi}{m^3}$.

By the property due to Lemma~\ref{lem:DKL},
\begin{align*}
\sum_{S \subseteq U} \lambda_S \Big( v_i(S\setminus T) - \bq(S) \Big) ~&\geq~ \sum_{S \subseteq U} \lambda_S \Big( v_i(S\setminus T) - \bq^\star(S) \Big) \\
&\geq~ \frac{1}{\alpha} \cdot v_i(U) - \bq^\star(T) \quad \geq \quad \frac{1}{\alpha} \cdot v_i(U) - 2 \bq(T) - \lvert U \rvert \frac{\psi}{m^3}.
\end{align*}
The claim follows because $\lvert U \rvert \leq m$.
\end{proof}

\subsection{Binary Search on the Candidate Prices}
\label{sec:binary-search-on-prices}

The mechanism takes $\psi$ as input, that is, the (estimated) value of $\max_i v_i(M)$. Based on this, it computes the set of candidate prices $B$ as described in Observation~\ref{obs:roundedPrices}\footnote{We assume that $\lvert B \rvert$ is a power of $2$, by adding arbitrary numbers to $B$ until reaching the closest power of $2$.}. It then determines a single price from $B$ for each item within $\beta = \log_2 \lvert B \rvert$ rounds. This is done akin to binary search. So in each of the $\beta$ rounds every item eliminates half of its candidate prices until it gets a unique price. This price is used for a fixed-price auction in the round $\beta+1$.


To perform the binary search, we partition the ordered set of bidders $N$ into $\beta + 1$ groups $N_1, \ldots, N_{\beta + 1}$. In round $\ell \in [\beta]$, for each item $e$, let $\mathbf{b}^{(\ell)}_e$ denote the ordered vector of remaining $\lvert B \rvert/2^{\ell-1}$ candidate prices. Let $p^{(\ell)}_e$ be the price exactly in the middle. We now try these prices in a fixed-price auction on the bidders in $N_\ell$. The outcome of this fixed-price auction is used to update the prices, that is, to determine $\mathbf{b}^{(\ell+1)}_e$ for all $e$. In particular, if one of the bidders in $N_\ell$ buys $e$, we define $\mathbf{b}^{(\ell+1)}_e$ to only contain prices from $\mathbf{b}^{(\ell)}_e$ that are at least $p^{(\ell)}_e$. Otherwise, it  only contain prices that are less than $p^{(\ell)}_e$.

While the above procedure might reach the correct price of an item, it can also end up with completely wrong prices. However, as we show, if we choose a random one of the $\beta+1$ fixed-prices auctions for the final allocation, the expected social welfare is sufficiently large.

\begin{tcolorbox}[breakable]
\underline{$\mech{(N,M,\psi)}$} 
\begin{enumerate}
	\item Let $B = \{0, 2^{-3 \lceil \log_2 m \rceil} \psi, 2^{1 - 3 \lceil \log_2 m \rceil} \psi, \ldots, 2^{-1} \psi\}$.
	\item Initialize $\mathbf{b}^{(1)}_e$ for all $e \in M$ to contain all elements of $B$ ordered increasingly.
	\item For every $i \in N$ draw $r_i$ independently uniformly from $[\beta+1]$, where $r_i$ denotes the round in which bidder $i$ participates. For any $\ell \in [\beta+1]$, let $N_\ell := \{ i \mid r_i = \ell\}$ in \underline{the same order} as $N$.
	\item Select a uniformly random final allocation round $r^\star \in [\beta+1]$.
	\item  For $\ell = 1$ to $\beta$ \underline{rounds}: 
	\begin{enumerate} [topsep=0pt, itemsep=0pt]
		\item For every item $e$, define the price $p^{(\ell)}_e = b^{(\ell)}_{e, k/2 + 1}$, where  $k =  \lvert B \rvert / 2^{\ell-1} $  and $\mathbf{b}^{(\ell)}_e = (b^{(\ell)}_{e, 1}, \ldots, b^{(\ell)}_{e, k})$ is the vector of  candidate prices for $e$;
		\item Run $\spmech(N_\ell,M,\bpricei{\ell})$ and let $\Ai{\ell}$ be the allocation;
		\item For every item $e$, if $e$ is allocated in $\Ai{\ell}$, define remaining candidate prices $\mathbf{b}_e^{(\ell+1)} = (b^{(\ell)}_{e, k/2+1}, \ldots, b^{(\ell)}_{e, k})$, and otherwise
		 define $\mathbf{b}_e^{(\ell+1)} = (b^{(\ell)}_{e, 1}, \ldots, b^{(\ell)}_{e, k/2})$;
		\item\label{line:coin-toss} If $r^\star = \ell$ then return $\Ai{\ell}$ as the final allocation. 
	\end{enumerate}
	\item Run $\spmech(N_{\beta+1},M,\bpricei{\beta+1})$, where $p^{(\beta+1)}_e$ in $\bpricei{\beta+1}$ is the unique price in  $\mathbf{b}^{(\beta+1)}_{e}$, and return this as the final allocation.
\end{enumerate}
\end{tcolorbox}

We remark that  $\mech$ can be implemented using one demand query per bidder. Its truthfulness is also easy to check as  bidders have no influence on their prices since they appear in at most one $\spmech$. 
The following theorem is our main technical result. 

\begin{theorem}\label{thm:main-mech}
	For a combinatorial auction with $n$ subadditive bidders and $m$ items, given $\psi \geq \max_i v_i(M)$, $\mech$ is universally truthful, uses at most $n$ demand queries and polynomial time, and has expected welfare 
	\[ \E[\Welfare]  ~\geq ~\frac{1}{2 \cdot \alpha \cdot (\beta+1)^2} \sum_i v_i(S_i^\star) - \frac{\psi}{m},
	\]
	where $\alpha = O(\log\!\log{m})$ for subadditive and $\alpha = O(1)$ for XOS valuations, and $\beta = O(\log\!\log{m})$. 
\end{theorem}

In \S\ref{sec:end-mech}, we will show that $\psi$ can be assumed to be order $ \max_i v_i(M)$ without loss of generality, which, together with Theorem~\ref{thm:main-mech}, immediately implies the desired $O((\log\!\log{m})^3)$-approximation for subadditive and 
$O((\log\!\log{m})^2)$-approximation for XOS (and submodular) valuations.

\section{The Analysis}

\newcommand{\correctbutsold}{D}


The proof of Theorem~\ref{thm:main-mech} will follow  from the following proposition.

\begin{proposition} \label{prop:main}
Consider any bidder $i \in N$ and the case that $\psi \geq v_i(S_i^\star)$. Let $r_{-i}$ denote a vector containing $r_{i'}$ for all $i' \neq i$. Then, for any choices of $r_{-i}$ we have
\[	\E_{r^\star, r_i}[\Util_i + \Rev(S_i^\star) \mid r_{-i} ] ~\geq~ \frac{1}{2 \cdot \alpha \cdot (\beta+1)^2} v_i(S_i^\star) - \frac{\psi}{m^2} \ ,
\]
where $\Util$ and $\Rev$ are defined for the $\spmech$ corresponding to the final allocation.
\end{proposition}

The proposition lower bounds for any bidder $i$ the sum of their utility  and the revenue generated from items in $S_i^\star$.  Before proving it, let us finish the proof of Theorem~\ref{thm:main-mech}.

\begin{proof} [Proof of Theorem~\ref{thm:main-mech}]
Note that Proposition~\ref{prop:main} holds for any fixed $r_{-i}$. So, in particular, we can also take the expectation over $r_{-i}$ to bound the unconditional expectation by
\[
\E\Big[ \Util_i + \Rev(S_i^\star) \Big] ~~=~~ \E_{r_{-i}}\E_{r^\star, r_i}\Big[ \Util_i + \Rev(S_i^\star) \;\Big|\; r_{-i} \Big] ~~\geq ~~\frac{1}{2 \cdot \alpha \cdot (\beta+1)^2} v_i(S_i^\star) - \frac{\psi}{m^2} \ .
\]

Since the item sets $S_i^\star$ are pairwise disjoint and because welfare equals total utility plus total revenue, we get
\[
\Welfare ~~=~~ \sum_i \left( \Util_i + \Rev(S_i^\star) \right) ~~\geq~~ \sum_{i: S_i^\star \neq \emptyset} \left( \Util_i + \Rev(S_i^\star) \right) \ ,
\]
where we use non-negativity of utilities. So, in combination, by linearity of expectation, we get
\begin{align*}
\E[\Welfare] ~~& \geq~~ \sum_{i: S_i^\star \neq \emptyset} \E\Big[ \Util_i + \Rev(S_i^\star) \Big]  \\
& \geq~~ \sum_{i: S_i^\star \neq \emptyset} \left( \frac{1}{2 \cdot \alpha \cdot (\beta+1)^2} v_i(S_i^\star) - \frac{\psi}{m^2} \right) ~~\geq~~ \frac{1}{2 \cdot \alpha \cdot (\beta+1)^2} \sum_i v_i(S_i^\star) - \frac{\psi}{m}\ ,
\end{align*}
where last inequality is because there are at most $m$ choices of $i$ for which $S_i^\star \neq \emptyset$. 
\end{proof}

\begin{proof} [Proof of Proposition~\ref{prop:main}]
For $e\in S_i^\star$, let $q_e$ denote the item prices given in Observation~\ref{obs:roundedPrices}, which we think of as the \emph{correct} price for $e$. 

Let us observe what happens when we run the $\mech$ on all bidders except bidder $i$. Recall that we have already conditioned on arbitrary values of $r_{-i}$. Further, if we condition  on $r^\star$, there is no more randomness in the decisions of $\mech{(N \setminus \{i\},M,\psi)}$.

For every item $e \in S_i^\star$, there are three possible scenarios depending on what happens to $e$ in the execution of $\mech{(N \setminus \{i\},M,\psi)}$, where we imagine running all the rounds even after the final allocation (alternately, think of what happens to $e$ for $ r^\star = \beta+1$). 
\begin{enumerate}[label=(\alph*), itemsep=0pt, topsep=0pt, leftmargin=20pt]
\item \label{case:Correct} $p^{(\beta+1)}_e = q_e$. That is, item $e$ is eventually offered at the correct price. 

\item \label{case:Oversold} $p^{(\beta+1)}_e > q_e$. That is, there is a level $\ell \in [\beta]$ for which $e$ is allocated in $A_\ell$ although $p^{(\ell)}_e > q_e$.

\item \label{case:Undersold} $p^{(\beta+1)}_e < q_e$. That is, there is a level $\ell \in [\beta]$ for which $e$ is not allocated in $A_\ell$ although $p^{(\ell)}_e \leq q_e$.
\end{enumerate}
We partition the items $S_i^\star$ based on which scenario takes place as follows.

Let $C \subseteq S_i^\star$ denote the set for which Scenario~\ref{case:Correct} takes place, i.e., these items are eventually \emph{correctly} priced with $p_e^{(\beta+1)} = q_e$. Among them, let $\correctbutsold \subseteq C$ be the items that are sold \emph{before} the arrival of bidder $i$ in $\spmech(N_{\beta+1},M,\bpricei{\beta+1})$, assuming $r_i = \beta+1 $. 

For the remaining items $S_i^\star \setminus C$, we partition them based on these scenarios and at which level the first ``incorrect'' event happened. Let $O_\ell \subseteq S_i^\star$ denote the set of \emph{oversold} items which are sold at a higher price in  Scenario~\ref{case:Oversold},  where $\ell$ is the smallest number such that $e$ is allocated in $A_\ell$ although $p^{(\ell)}_e > q_e$. 
Analogously, let $U_\ell \subseteq S_i^\star$ denote the set of \emph{undersold} items $e$ that in Scenario~\ref{case:Undersold} in level $\ell$. 
Note that all of the above sets are disjoint and their union is $S_i^\star$, i.e., $S_i^\star = \cupdot_{\ell=1}^\beta ( U_\ell  \cupdot O_\ell)\cupdot C$.

The general structure of the proof is as follows. Let $T = \cupdot_{\ell=1}^\beta O_\ell \cupdot \correctbutsold$ denote the items that are sold on some level at a sufficiently high price. So, we will use their contribution to the revenue. For the remaining items in $S_i^\star \setminus T$ in contrast, we will argue that bidder $i$ could buy them at some point (depending on randomness $r_i$ and $r^\star$). So, we will use their contribution to bidder $i$'s utility.

We first lower bound the expected revenue $\E[\Rev(S_i^\star)]$.

\begin{claim}\label{claim:RevLB}
For any choice of $r_{-i}$, we have
\[
\E_{r^\star, r_i}[\Rev(S_i^\star) \mid r_{-i}] ~\geq~ \frac{1}{(\beta+1)^2} \cdot \bq(T)\ .
\]
\end{claim}

\begin{proof}
We will lower-bound the revenue for all possible values of $r_i$ and $r^\star$. If $r_i \neq \beta+1$, we use the trivial bound of $\Rev(S_i^\star) = 0$.

Next, if $r_i = \beta+1$ and $r^\star = \ell$ for any $\ell < \beta+1$, we know that that $A_\ell$ is the final allocation. Moreover, 
the set $O_{\ell}$ defined in $\mech{(N \setminus \{i\},M,\psi)}$ remains the same even in the actual mechanism $\mech{(N,M,\psi)}$, as in the latter, bidder $i$ ``appears'' after round $\ell$. 
In this allocation, each item $e \in O_\ell$ is sold at a price of at least $q_e$. So $\Rev(S_i^\star) \geq \bq(O_\ell)$. 

Finally, if $r_i = \beta+1 = r^\star$, the  allocation is $A_{\beta+1}$. Here, each item  $e \in \correctbutsold$ is sold at a price at least $q_e$, the same way in 
$\mech{(N \setminus \{i\},M,\psi)}$ and $\mech{(N,M,\psi)}$, as in the latter, bidder $i$ ``appears'' after the items in $\correctbutsold$ are sold. 
So $\Rev(S_i^\star) \geq \bq(\correctbutsold)$. 

Since every combination of $r_i$ and $r^\star$ happens with probability $\frac{1}{(\beta+1)^2}$, we have
\begin{align*}	\label{eq:revLB}
\E_{r^\star, r_i}[\Rev(S_i^\star) \mid r_{-i}] ~~\geq ~~ \frac{1}{(\beta+1)^2} \Big(  \sum_{\ell=1}^\beta \bq(O_\ell) + \bq(\correctbutsold) \Big) ~~=~~ \frac{1}{(\beta+1)^2} \cdot \bq(T) \ ,
\end{align*}
where the last equality uses that $T$ is defined to be $\cupdot_{\ell=1}^\beta O_\ell \cupdot \correctbutsold$.
\end{proof}
 
 Next we lower bound bidder $i$'s utility for an arbitrary set $S \subseteq S_i^\star$ (the reason we bound this for all subsets of $S_i^\star$ and not only $S_i^\star$ itself is because the prices in Observation~\ref{obs:roundedPrices} and Lemma~\ref{lem:DKL} are accompanied 
 by a distribution $\lambda$ which may put different masses on different subsets of $S_i^\star$).
 
\begin{claim} \label{claim:UtilLB}
For any choice of $r_{-i}$, and any arbitrary set $S \subseteq S_i^\star$, we have
\[
\E_{r^\star, r_i}[\Util_i \mid r_{-i}] ~\geq~ \frac{1}{(\beta+1)^2} \cdot \Big(   v_i(  S \setminus T) - \bq( S \setminus T) \Big) \ .
\]
\end{claim}

\begin{proof}
To lower-bound the LHS, we will again consider all possible choices of $r_i$ and $r^\star$. First note that if $r_i \neq r^\star$, then $\Util_i = 0$ because bidder $i$ does not get any items in the final allocation.

Next, consider the case where $r_i = r^\star= \ell < \beta + 1$. In this case, bidder $i$ can purchase items in the set $U_\ell \cap S$ by paying price at most $\bq(U_\ell \cap S)$ since these items are not sold in this round in the absence of bidder $i$. Thus, in this case, the utility is lower-bounded by
\[ \Util_i ~\geq~ v_i(U_\ell \cap S) - \bq(U_\ell \cap S) \ .
\]

Now consider the case where $r_i = \beta+1 = r^\star$. Since bidder $i$ can purchase items in the set $(C \setminus \correctbutsold) \cap S$ by paying prices at most $\bq((C \setminus \correctbutsold) \cap S)$, the utility is this case is at least
\[ \Util_i ~\geq ~ v_i ((C \setminus \correctbutsold) \cap S) - \bq((C \setminus \correctbutsold) \cap S) \ .
\]

Any combination of $r_i$ and $r^\star$ happens with probability $\frac{1}{(\beta+1)^2}$. Therefore, we get
\begin{align*} \E_{r^\star, r_i}[\Util_i \mid r_{-i}]  &\geq \frac{1}{(\beta+1)^2} \Big(  \sum_{\ell=1}^{\beta} v_i(U_\ell \cap S)  +v_i ((C \cap S) \setminus \correctbutsold) - \bq((C \cap S) \setminus \correctbutsold) - \sum_{\ell=1}^{\beta} \bq(U_\ell \cap S) \Big) \\
&\geq \frac{1}{(\beta+1)^2} \Big(   v_i( S \setminus T) - \bq( S \setminus T) \Big) \ ,
\end{align*}
where the last inequality uses subadditivity of $v_i$ and that $S_i^\star = \cupdot_{\ell=1}^\beta ( U_\ell  \cupdot O_\ell)\cupdot C$.
\end{proof}

Since $\Util_i \geq 0$ pointwise, we have $\E_{r^\star, r_i}[\Util_i \mid r_{-i}] \geq \frac{1}{2} \E_{r^\star, r_i}[\Util_i \mid r_{-i}]$ and so an immediate implication of Claim~\ref{claim:UtilLB} is that 
$\E_{r^\star, r_i}[\Util_i \mid r_{-i}] \geq \frac{1}{2 \cdot (\beta+1)^2} \Big(   v_i(  S \setminus T) - \bq( S \setminus T) \Big)$. Combining this with Claim~\ref{claim:RevLB}, we can lower bound the expected sum of revenue and utility by
\begin{align}	 \label{eq:RevAndUtil}
\E_{r^\star, r_i}[\Util_i + \Rev(S_i^\star) \mid r_{-i}]  ~\geq~ \frac{1}{2 \cdot (\beta+1)^2} \Big(  v_i( S \setminus T) - \bq( S \setminus T) \Big) +   \frac{\bq(T)} {(\beta+1)^2} \ ,
\end{align}
for an arbitrary subset of  items $S \subseteq S_i^\star$. Finally, we use Observation~\ref{obs:roundedPrices} to get
\[ \sum_{S \subseteq S_i^\star} \lambda_S \Big( v_i(S\setminus T) - \bq(S \setminus T) \Big) ~\geq ~  \sum_{S \subseteq S_i^\star} \lambda_S \Big( v_i(S\setminus T) - \bq(S) \Big) ~\geq~ \frac{1}{\alpha} \cdot v_i(S_i^\star) - 2 \cdot \bq(T) - \frac{\psi}{m^2} \ .
\]
Combining this with Eq.~\eqref{eq:RevAndUtil} and using $\sum_{S\subseteq S_i^\star} \lambda_S \leq 1$, we get
\begin{align*}
\E_{r^\star, r_i}[\Util_i + \Rev(S_i^\star) \mid r_{-i}] ~&\geq~ \frac{1}{2 \cdot (\beta+1)^2} \sum_{S \subseteq S_i^\star} \lambda_S\Big(  v_i(  S \setminus T) - \bq( S \setminus T) \Big)  + \frac{ \bq(T)}{(\beta+1)^2} \\
&\geq~ \frac{1}{2 \cdot (\beta+1)^2}\cdot  \paren{\frac{1}{\alpha} \cdot v_i(S_i^\star) - 2 \cdot \bq(T) - \frac{\psi}{m^2}} + \frac{ \bq(T)}{(\beta+1)^2} \\
&\geq~ \frac{1}{2 \cdot (\beta+1)^2}\cdot  \frac{1}{\alpha} \cdot v_i(S_i^\star) - \frac{\psi}{m^2} \ ,
\end{align*}
which completes the proof of this proposition.
\end{proof}


\newcommand{\Nsecprice}{\ensuremath{N_{\textnormal{\textsf{second price}}}}}
\newcommand{\Nstat}{\ensuremath{N_{\textnormal{\textsf{stat}}}}}
\newcommand{\OPTstat}{\ensuremath{\OPT_{\textnormal{\textsf{stat}}}}}
\newcommand{\ALGstat}{\ensuremath{\alg_{\textnormal{\textsf{stat}}}}}
\newcommand{\Nmech}{\ensuremath{N_{\textnormal{\textsf{mech}}}}}
\newcommand{\OPTmech}{\ensuremath{\OPT_{\textnormal{\textsf{mech}}}}}
\newcommand{\ALGmech}{\ensuremath{\alg_{\textnormal{\textsf{mech}}}}}

\newcommand{\fmech}{\ensuremath{\textnormal{\textsf{FinalMechanism}}}\xspace}

\section{Removing the Extra Assumption}\label{sec:end-mech}

We now remove the assumption on the knowledge of an upper bound $\psi$ for $\max_{i \in N} v_i(M)$ and obtain our final mechanism. 
Simply asking the bidders for this value would lead to complex incentive issues, and thus, we instead apply  a common sampling approach~\cite{DobzinskiNS06,Dobzinski07,Dobzinski16,AS-FOCS19}.

\begin{theorem} \label{thm:finalMain}
	There is a universally truthful mechanism for combinatorial auctions with at most $n$ demand and value queries and polynomial time that achieves an 
	$O((\log\!\log{m})^3)$-approximation for subadditive and $O((\log\!\log{m})^2)$-approximation for XOS (and submodular) bidders, in expectation. 
\end{theorem}
\begin{proof}
	The mechanism is as follows. We first independently add each bidder to the set $\Nsecprice$ with probability $\frac{1}{2}$. Then, we run a second-price auction, in which the grand bundle as a whole is assigned to the bidder in $\Nsecprice$ whose valuation $v_i(M)$ is highest (using value queries). With probability $\frac{1}{2}$, this is the final allocation. Otherwise, we let $\psi := \max_{i \in \Nsecprice} v_i(M)$, and run $\mech{(N \setminus \Nsecprice,M,\psi)}$ and allocate the items to
	bidders in $N \setminus \Nsecprice$ accordingly. This mechanism is  universally truthful and uses for each bidder only one demand or value query and polynomial time. Thus, it only remains to analyze its approximation ratio.

Let $i_{\mathrm{max}}$ be a bidder $i$ whose valuation for the grand bundle $v_i(M)$ is maximum. In the following, we condition on $i_{\mathrm{max}} \in \Nsecprice$, which happens with probability $\frac{1}{2}$. If $i_{\mathrm{max}} \not\in \Nsecprice$, we use that the social welfare is non-negative.

Due to our conditioning, the social welfare obtained in the second-price auction is always $\max_i v_i(M) = v_{i_{\mathrm{max}}}(M)$. Besides, we always have $\psi = \max_i v_i(M)$. Therefore, by Theorem~\ref{thm:main-mech}, the expected social welfare obtained by $\mech{(N \setminus \Nsecprice,M,\psi)}$ for a fixed $\Nsecprice$ (which contains $i_{\mathrm{max}}$) is at least
\[
\sum_{i \in N \setminus \Nsecprice} \frac{1}{2 \cdot \alpha \cdot (\beta+1)^2} \cdot v_i(S_i^\star) - \frac{\psi}{m} ~~=~~ \sum_{i \in N \setminus \Nsecprice} \frac{1}{2 \cdot \alpha \cdot (\beta+1)^2} \cdot v_i(S_i^\star) - \frac{v_{i_{\mathrm{max}}}(M)}{m} \enspace.
\]

Even when conditioning on $i_{\mathrm{max}} \in \Nsecprice$, every $i \neq i_{\mathrm{max}}$ is contained in $\Nsecprice$ with probability $\frac{1}{2}$. Therefore, taking the conditional expectation over $\Nsecprice$, we get the following lower bound on the expected social welfare obtained by $\mech$:
\[
\frac{1}{4 \cdot \alpha \cdot (\beta+1)^2} \sum_{i \in N \setminus \{i_{\mathrm{max}}\}} v_i(S_i^\star) - \frac{v_{i_{\mathrm{max}}}(M)}{m} \enspace.
\]

Finally, the second-price auction and $\mech$ are each run with probability $\frac{1}{2}$. So, still conditioning on $i_{\mathrm{max}} \in \Nsecprice$, the expected welfare of the final mechanism is at least
\[
\frac{1}{2} v_{i_{\mathrm{max}}}(M) + \frac{1}{2} \left( \frac{1}{4 \cdot \alpha \cdot (\beta+1)^2} \sum_{i \in N \setminus \{i_{\mathrm{max}}\}} v_i(S_i^\star) - \frac{v_{i_{\mathrm{max}}}(M)}{m} \right) ~~\geq~~ \frac{1}{8 \cdot \alpha \cdot (\beta+1)^2} \sum_{i \in N} v_i(S_i^\star) \enspace.
\]
As $i_{\mathrm{max}} \in \Nsecprice$ with probability $\frac{1}{2}$, the expected welfare of the mechanism is at least half the above bound. 
As $\beta = O(\log\!\log{m})$ and $\alpha = O(\log\!\log{m})$ for subadditive and $\alpha = O(1)$ for XOS bidders, we obtain the final bound in the theorem. 
\end{proof}


\section{Concluding Remarks and Open Problems}\label{sec:conc}

We gave a computationally-efficient and universally truthful mechanism for subadditive combinatorial auctions  that achieves an $O((\log\!\log{m})^3)$-approximation in expectation.  
This breaks the longstanding logarithmic barrier on the approximation ratio of such mechanisms for subadditive bidders dating back to the $O(\log{m} \cdot \log\!\log{m})$ approximation
mechanism of~\cite{Dobzinski07}. Along the way, we also considerably simplified the previously best mechanism of~\cite{AS-FOCS19} for submodular (and XOS) bidders, and improved its approximation ratio from $O((\log\!\log{m})^3)$ to $O((\log\!\log{m})^2)$. 

The obvious question left open by our work is whether our bounds  can be  improved further.  The limit of the recent ``price learning'' approaches for this problem on submodular or subadditive bidders in~\cite{Dobzinski16,AS-FOCS19} and our work seems to be
 an $\Omega(\log\!\log{m})$ approximation; it is an interesting open question if we can improve our mechanisms to match this limit. However, the most fascinating open question in this line of work is 
 to improve the approximation factor all the way down to a constant or alternatively prove that this is not possible. At this point, we do not even have a lower bound that rule out such possibility for deterministic mechanisms 
 even though the best deterministic mechanism for this problem only achieves an $O(\sqrt{m})$ approximation~\cite{DobzinskiNS06}. Thus, improving the lower bound techniques 
 for this problem is another very interesting open question. 

{\small
\bibliographystyle{alpha}
\bibliography{bib}
}

\newpage

\appendix

\IGNORE{
\medskip
\noindent
{\bf Acknowledgments}.
We thank a number of colleagues for useful discussions.
We are grateful to Sylvia Boyd and Paul Elliott-Magwod
for information on ATSP integrality gaps.
}


\IGNORE{
}

\end{document}